\documentclass[a4paper,superscriptaddress,twocolumn,nofootinbib,accepted=2020-04-20]{quantumarticle}
\pdfoutput=1
\usepackage[utf8]{inputenc}
\usepackage[english]{babel}
\usepackage[T1]{fontenc}
\usepackage{amsmath}
\usepackage{hyperref}
\usepackage[numbers,sort&compress]{natbib}

\usepackage{tikz}
\usepackage{amsmath}
\usepackage{amssymb}
\usepackage{epsfig}
\usepackage{color,epstopdf}
\usepackage{amsthm, thm-restate}
\usepackage{hyperref}
\usepackage{subcaption}
\usepackage{hyperref,amsthm,mathtools,braket}
\hypersetup{colorlinks=true,linkcolor=blue,citecolor=blue}

\DeclareMathOperator{\Tr}{Tr}

\newcommand{\ketbra}[2]{\ensuremath{\left|#1\right\rangle\!\left\langle#2\right|}}
\renewcommand{\braket}[2]{\ensuremath{\!\left\langle#1 \vphantom{#2}\middle| #2 \vphantom{#1}\right\rangle \!}}

\DeclareSymbolFont{bbold}{U}{bbold}{m}{n}
\DeclareSymbolFontAlphabet{\mathbbold}{bbold}
\newcommand{\iden}{\mathbbold{1}}

\newcommand{\chn}{\mathcal{D}}
\newcommand{\NN}{\mathcal{N}}
\newcommand{\UU}{\mathcal{U}}
\newcommand{\Id}{\mathbb{I}}

\makeatletter
\newcommand*{\inlineequation}[2][]{%
  \begingroup
    \refstepcounter{equation}%
    \ifx\\#1\\%
    \else
      \label{#1}%
    \fi
    \relpenalty=10000 %
    \binoppenalty=10000 %
    \ensuremath{%
      #2%
    }%
    ~\@eqnnum
  \endgroup
}
\makeatother

\theoremstyle{plain}
\newtheorem{thm}{Theorem}
\newtheorem{lem}[thm]{Lemma}

\theoremstyle{definition}

\theoremstyle{remark}

\begin{document}

\title{Contextual advantage for state-dependent cloning}
\author{Matteo Lostaglio}
\affiliation{ICFO-Institut de Ciencies Fotoniques, The Barcelona Institute of Science and Technology, Castelldefels (Barcelona), 08860, Spain}
\affiliation{QuTech, Delft University of Technology, P.O. Box 5046, 2600 GA Delft, The Netherlands}
\author{Gabriel Senno}
\affiliation{ICFO-Institut de Ciencies Fotoniques, The Barcelona Institute of Science and Technology, Castelldefels (Barcelona), 08860, Spain}

\begin{abstract}
A number of noncontextual models exist which reproduce different subsets of quantum theory and admit a no-cloning theorem. Therefore, if one chooses noncontextuality as one's notion of classicality, no-cloning cannot be regarded as a nonclassical phenomenon. In this work, however, we show that there are aspects of the phenomenology of quantum state cloning which are indeed nonclassical according to this principle. Specifically, we focus on the task of state-dependent cloning and prove that the optimal cloning fidelity predicted by quantum theory cannot be explained by any noncontextual model. We derive a noise-robust noncontextuality inequality whose violation by quantum theory not only implies a quantum advantage for the task of state-dependent cloning relative to noncontextual models, but also provides an experimental witness of noncontextuality.
\end{abstract}

\maketitle

An important guiding principle for quantum theorists is the identification of genuine nonclassical effects certified by rigorous theorems. Given a quantum phenomenon, the relevant question is: Are there classical models able to reproduce the observed operational data? Here we investigate this question in the context of a cloning experiment.  

The no-cloning theorem \cite{wootters1982single, dieks1982communication, park1970concept} is widely regarded as a central result in quantum theory. Informally, the theorem states the impossibility of copying quantum information, and is contrasted with the fact that classical information, on the other hand, can be perfectly copied. More precisely, there is no machine (formally, a quantum channel) that can take two distinct and nonorthogonal states $\{\ket{\psi_1},\ket{\psi_2}\}$ sent at random as inputs and output the corresponding copies $ \{\ket{\psi_1}\otimes \ket{\psi_1},\ket{\psi_2}\otimes \ket{\psi_2}\}$ \cite{yuen1986amplification}.

While no-cloning is often regarded as an intrinsically quantum feature, one would like to back that claim by a precise theorem stating what operational features cannot be explained within classical models. The theorem should hence define a precise notion of
`classicality' and show that such notion leads to operational predictions incompatible with the relevant quantum statistics \cite{schmid2018contextual}.  
At the operational level, we can schematically think of an experiment as a set of black-boxes each corresponding to certain sets of operational instructions. At the ontological level we look for theoretical explanations of the empirical data within the framework of \emph{ontological models}. This is a very broad class of models involving an arbitrary set of physical states evolving according to some laws and ultimately determining (the probabilities of) the measurement outcomes. This analysis forces us to look for any plausible alternative explanation of the empirical data collected in a quantum experiment before we certify it as ``nonclassical''. But, which ontological models should be deemed ``classical"?

Clearly, the
broader the chosen notion of classicality is, the stronger the resulting no-go theorem is.  Since the scenario of quantum cloning does not feature space-like separated measurements, we need a different notion of `classicality' than the ubiquitous Bell's locality. Hence, in this work we identify nonclassical features as those that cannot be explained within any noncontextual model, in the generalized sense introduced in Ref.~\cite{spekkens2005contextuality}. It is a known fact that, with respect to this broad notion, no-cloning by itself should not be regarded as a nonclassical phenomenon. There are, in fact, several examples of noncontextual models for subsets of quantum theory with a no-cloning theorem \cite{bartlett2012reconstruction,spekkens2007evidence}. The mechanism behind no-cloning in noncontextual theories is simple: non-orthogonal quantum states $\ket{\psi_1}$, $\ket{\psi_2}$ correspond to overlapping probability distributions $\mu_1(\lambda)$, $\mu_2(\lambda)$ over the posited set of physical states $\lambda$ and there is no deterministic nor stochastic process mapping $\{\mu_1, \mu_2\}$ to $\{\mu_1 \otimes \mu_1, \mu_2 \otimes \mu_2\}$ \cite{daffertshofer2002classical}. The existence of these models proves that no-cloning cannot be interpreted as a nonclassical phenomenon when the notion of classicality is taken to be that of noncontextuality.\footnote{Crucially, cloning should be distinguished from the notion of broadcasting. Broadcasting only requires the creation of a joint distribution with marginals $\mu_i$ and can be done perfectly by a generalized CNOT.} 
Hence, we need to look more closely at the phenomenology of quantum cloning if we are to identify aspects of it that are nonclassical according to the principle of noncontextuality.

In this work, we identify a strongly nonclassical aspect in the ultimate limits of imperfect cloning. The question of what is the best fidelity with which a given set of quantum states can be cloned has been widely studied since the pivotal work of Bu\v{z}ek and Hillary in 1996 \cite{buvzek1996quantum} (for a review on quantum cloning, see, \emph{e.g.}, Ref.~\cite{scarani2005cloning}). 
We find that the optimal fidelity predicted
by quantum theory for the cloning of two distinct, non orthogonal pure states cannot be reproduced by any noncontextual model which complies with the operational phenomenology featured in a quantum
cloning experiment.
 Specifically, contextuality provides an advantage to the maximum copying fidelity.
Our result
directly links contextuality to a quantum advantage \cite{schmid2018contextual,saha2019state,tavakoli2020measurement}.

\section{Noncontextual ontological models of operational theories} 

At the operational level, we can schematically think of an experiment as a set of black-boxes each corresponding to certain sets of operational instructions.\footnote{While empirical data is always to some degree \emph{theory-laden}, the word ``operational'' here signifies that we are striving towards the ideal of the most low-level instructions we can imagine (e.g. press this button, write down an outcome when a corresponding light flashes etc.). This is to be opposed to high-level instructions that refer to theoretical entities, such as ``lower the potential barrier in which the electron is trapped''.} We can distinguish three kinds of black-boxes:
	\begin{enumerate}
		\item A preparation black-box $P_s$ initialises the system;
		\item A transformation black-box $T$ takes in a system prepared according to $P_s$ and transforms it into some new preparation, denoted $T(P_s)$.
		\item  A measurement black-box $M_{s'}$ takes a preparation $P_{s}$ as input and returns an outcome $x$ with probability $p(x|P_{s},M_{s'})$.
		\item An experiment consists of collecting the statistics $p(x|T(P_s),M_{s'})$ for various choices of the black boxes $P_s$, $T$ and $M_{s'}$.
	\end{enumerate}

 The set of $P_s$, $T$, $M_{s'}$ and corresponding observed statistics $p(x|T(P_s),M_{s'})$ are the defining elements of an \emph{operational theory}. Noncontextuality is a restriction on the ontological models that try to explain the statistics of some operational theory. An ontological model for an operational theory is one which \cite{leifer2014quantum}:

\begin{enumerate}
\item Makes every preparation $P_s$ correspond to sampling from a probability distribution $\mu_{s}(\lambda)$ over some set of ontic variables $\lambda$. $\lambda$s are referred to as `hidden variables' in the context of Bell nonlocality and they form a (measurable) set $\Lambda$.
\item Represents transformations by matrices $T(\lambda'|\lambda)$ of transition probabilities ($T(\lambda'|\lambda) \geq 0$, $\int d\lambda' T(\lambda'|\lambda) =1$ $\forall \lambda$) acting on the corresponding probability density.
\item Represents a measurement  $M_{s'}$ by a response function  $\xi_{s'}(x|\lambda)$  giving the probability of outcome $x$ given that the hidden variable takes the value $\lambda$ (\mbox{$\xi_{s'}(x|\lambda)\geq 0$}, \mbox{$\sum_x \xi_{s'}(x|\lambda)=1$ $\forall \lambda$}).
\end{enumerate}
An ontological model then defines its predictions as
\begin{equation}\label{eq:probability-rule}
p(x|T(P_s),M_{s'}) = \int d\lambda d\lambda' \mu_s(\lambda) T(\lambda'|\lambda )\xi_{s'}(x|\lambda').
\end{equation}

Two operational procedures (be them preparations, measurements or transformations) are said to be \emph{operationally equivalent} if they cannot be distinguished by any experiment. Noncontextuality, in the generalized form introduced in \cite{spekkens2005contextuality}, is a restriction to ontological models requiring that \emph{if two procedures are operationally equivalent, they must be represented by the same object in the ontological model}. This notion can be seen as an extension of the traditional one of Kochen-Specker~\cite{kochen1975problem,spekkens2005contextuality}. 

In this work we will be concerned with operational equivalences only at the level of preparations. Two preparations $P_s$ and $P_{s'}$ are operationally equivalent if they cannot be distinguished by any measurements:
\begin{equation*}
p(x|P_s,M) = p(x|P_{s'},M), \quad \forall M,
\end{equation*}
which, for short, we will denote by $P_s \simeq P_{s'}$. The assumption of (preparation) noncontextuality is then
\begin{equation}
\label{eq:preparationNC}
P_s \simeq P_{s'} \Rightarrow \mu_s(\lambda) = \mu_s(\lambda').
\end{equation}
This principle can be understood as an `identity of the indiscernibles' and, together with locality, it can be seen as a successful methodological principle for theory construction \cite{spekkens2019ontological}. Examples of noncontextual ontological models include classical Hamiltonian mechanics, Hamiltonian mechanics with a resolution limit on phase space \cite{bartlett2012reconstruction} and Spekken's toy model \cite{spekkens2007evidence}.

\begin{figure}[t]
	\centering
	\includegraphics[width=0.8\linewidth]{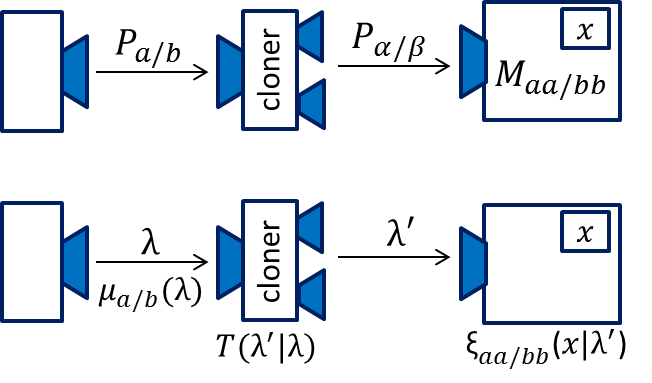}
	\caption{\emph{Cloning experiment.} Top: black-box of the cloning protocol; one of two preparation procedures $P_x$, $x=a,b$ is performed with equal probability, the resultant state is sent through a cloning machine (independent of $x$), which respectively prepares $P_\gamma$, $\gamma=\alpha,\beta$; a test measurement $M_{xx}$ for the target preparation $P_{xx}$ is performed and passed with probability $P(M_{aa}|P_\alpha)$ (or $P(M_{bb}|P_\beta)$). Bottom: ontological description of the same experiment, where preparing $P_x$ corresponds to sampling $\lambda$ with probability $\mu_x(\lambda)$, the cloning machine maps $\lambda \mapsto \lambda'$ with probability $T(\lambda'|\lambda)$ and $M_{xx}$ gives a `pass' outcome with probability $\xi_{xx}(1|\lambda')$.} 
	\label{fig:scheme}
\end{figure}

\section{Operational features of quantum cloning - ideal scenario}

We now describe the operational features of optimal state-dependent quantum cloning which, as we will show, are impossible to explain with noncontextual models (see also Fig.~\ref{fig:scheme}). We will make the assumption that certain perfect correlations are observed, but we will later remove these idealizations. For all two-outcome measurements $M_s$ we will use the shortcuts $p(x=1|P,M_s) \equiv p(M_s|P)$ and $\xi_{s}(1|\lambda')\equiv \xi_{s}(\lambda')$.

Let $P_a$ and $P_b$ denote the experimental procedures followed to prepare the states $\ket{a}$ and $\ket{b}$ to be cloned. As an operational signature of the fact that $\ket{a}$ and $\ket{b}$ are two pure and, in general, nonorthogonal states, we consider the `test measurements' $M_a$, $M_b$, with outcomes  $x \in \{0,1\}$, giving the operational statistics $p(M_a|P_a) = p(M_b|P_b) = 1$.  In the quantum formalism, this statistics is reproduced by performing the projective measurements $\{ \ketbra{a}{a}, \iden - \ketbra{a}{a}\}$ and $ \{ \ketbra{b}{b}, \iden - \ketbra{b}{b}\}$ (with $x=1$ corresponding to the first outcome). We will use the notation $c_{ab} := p(M_{b}|P_{a})$, which is called `confusability' in Ref.~\cite{schmid2018contextual}, for the probability of observing the first outcome of the $M_b$ measurement when the system is initialized according to $P_a$. Clearly, in the ideal quantum experiment one observes $c_{ab} = |\braket{a}{b}|^2$. 

The two preparations $P_a$, $P_b$ go through a cloning machine $T$, which outputs new preparations $P_\alpha = T(P_a)$, $P_\beta=T(P_b)$. In quantum theory, the optimal-state dependent cloning operation is a unitary $U$ and, hence, the preparations $P_\alpha$ and $P_\beta$ correspond to pure states $\ket{\alpha} := U \ket{a \, 0}$, $\ket{\beta} := U \ket{b \, 0}$ respectively, with $\ket{0}$ the initial state of some ancillary register. Operationally, and similarly to the discussion above, the purity of the outputs implies that we can perform test measurements $M_{\alpha}$, $M_{\beta}$ satisfying $p(M_\alpha|P_\alpha) = 1$, \mbox{$p(M_\beta|P_\beta) = 1$} (again, by performing the measurements described in the quantum formalism as $\{ \ketbra{\alpha}{\alpha}, \iden - \ketbra{\alpha}{\alpha}\}$ and $ \{ \ketbra{\beta}{\beta}, \iden - \ketbra{\beta}{\beta}\}$). 

The experiment ends by testing what the fidelity between the output and the ideal clone is. To do so, given the ideal clones $P_{aa}$, $P_{bb}$ we introduce test-measurements $M_{aa}$, $M_{bb}$ and assume one observes the statistics $p(M_{aa}|P_{aa}) = p(M_{bb}|P_{bb}) = 1$, $p(M_{bb}|P_{aa}) = |\braket{aa}{bb}|^2 = |\braket{a}{b}|^4.$ In a quantum experiment this is realized by preparing states $\ket{aa}$, $\ket{bb}$ and performing the projective measurements $\{ \ketbra{aa}{aa}, \iden - \ketbra{aa}{aa}\}$, $ \{ \ketbra{bb}{bb}, \iden - \ketbra{bb}{bb}\}$.

Then, denoting by $c_{\alpha aa} := P(M_{aa}|P_{\alpha})$, $c_{\beta bb} := P(M_{bb}|P_{\beta})$,  the (global) cloning fidelity is operationally defined to be $$F_g:=\frac{1}{2}c_{\alpha aa} + \frac{1}{2} c_{\beta bb},$$ i.e., the average probability that the imperfect clones $P_\alpha$ and $P_\beta$ pass the corresponding test measurements for the ideal clones, $M_{aa}$ and $M_{bb}$ respectively. In quantum theory, the optimal cloning unitary achieves~\cite{bruss1998optimal}
\begin{align}
\label{eq:quantumbound}
F^{\rm Q,opt}_{g} &:= \frac{1}{4} \left[ \sqrt{(1+c_{ab})(1+ \sqrt{c_{ab}})} \right. \nonumber \\
&\qquad\qquad \left.+ \sqrt{(1-c_{ab})(1-\sqrt{c_{ab}})}\right]^2,
\end{align}
with $c_{ab} = | \braket{a}{b}|^2$.

This brief summary captures the main operational features 
of the traditional `optimal state-dependent cloning' and highlights the main issue with this
approach: it leaves no room to leverage operational equivalences to further study its  potential non-classical
aspects. To fix that, we follow Ref.~\cite{schmid2018contextual} and exploit another operational consequence of the purity of $\ket{a}$, $\ket{b}$: the existence of preparations $P_{a^\perp}$, $P_{b^\perp}$ satisfying $p(M_a|P_{a^\perp}) =  p(M_b|P_{b^\perp}) = 0$ and such that the mixture $ P_a/2 + P_{a^\perp}/2$ (tossing a fair coin and following either $P_a$ or $P_{a^\perp}$) is operationally equivalent to the mixture $ P_{b}/2 + P_{b^{\perp}}/2$: $P_a/2 + P_{a^\perp}/2 \simeq  P_{b}/2 + P_{b^{\perp}}/2$.  In the idealized quantum experiment one observes this operational statistics by preparing pure states $\ket{a^\perp}$, $\ket{b^\perp}$ in the span of \{$\ket{a}$, $\ket{b}$\} and satisfying $\braket{a}{a^\perp} = \braket{b}{b^\perp}=0$ as well as $\frac{1}{2} \ketbra{a}{a} +  \frac{1}{2} \ketbra{a^\perp}{a^\perp} = \frac{1}{2} \ketbra{b}{b} +  \frac{1}{2} \ketbra{b^\perp}{b^\perp}$. The same discussion can be repeated for each of the pairs $\{(a,b), (\alpha, aa), (\beta, bb) \}$.

To conclude, here is an operational account  (without any reference to quantum theory) of the features that we demand are observed in the idealized scenario of the cloning experiment: there exists $ P_s, P_{s^\perp}, M_s$ such that 
\begin{enumerate}
	\item[O1] $p(M_{s}|P_s) = 1$, $p(M_s|P_{s^\perp}) = 0$ for $s = a,b,\alpha,\beta,aa,bb$.
	\item[O2] $\frac{1}{2} P_s + \frac{1}{2} P_{s^\perp} \simeq \frac{1}{2} P_{s'}+ \frac{1}{2} P_{s^{'\perp}}$, for all $(s,s')$ in $\{(a,b), (\alpha, aa), (\beta, bb) \}$.
\end{enumerate}

\section{Optimal cloning is contextual - ideal scenario}

In any ontological model, a cloning experiment is described as follows (see Fig.~\ref{fig:scheme}). A preparation device randomly prepares either $P_a$ or $P_b$, i.e., it samples a $\lambda$ from either the distribution $\mu_a(\lambda)$ or $\mu_b(\lambda)$. This state is sent into the cloning machine that maps $\lambda$ into some new $\lambda'$ with probability $T(\lambda'|\lambda)$. For example, if $\lambda = (x_1,p_1)$ one could have $\lambda' = (x'_1,p'_1,x'_2,p'_2)$. This $\lambda'$ is sent into a testing device doing the measurement $M_{aa}$ if $P_a$ was prepared, or $M_{bb}$ if $P_b$ was prepared. Upon receiving $\lambda'$, the device gives an outcome $x$ with probability $\xi_{aa}(\lambda')$ or $\xi_{bb}(\lambda')$.

The assumption of noncontextuality (more precisely, preparation noncontextuality \cite{spekkens2005contextuality}) and linearity applied to the operational equivalences in O2 requires that  any noncontextual ontological model must satisfy (see Eq.~\eqref{eq:preparationNC})
\begin{equation}
\label{eq:noncontextualityconsequence}
\frac{1}{2} \mu_s(\lambda) + \frac{1}{2} \mu_{s^\perp}(\lambda) = \frac{1}{2} \mu_{s'}(\lambda) + \frac{1}{2} \mu_{s^{'\perp}}(\lambda),
\end{equation}
for all $(s,s')$ in $\{(a,b), (\alpha, aa), (\beta, bb) \}$ and $\lambda \in \Lambda$. Our main result is that no noncontextual ontological model can reproduce the operational features listed O1-O2 and match the optimal cloning fidelity predicted by quantum theory. More precisely:

\begin{thm}[Optimal cloning fidelity in noncontextual models]
\label{thm:cloningcontextual-ideal}
Let $P_\alpha=T(P_a)$, $P_{\beta}=T(P_a)$ be the achieved outputs of a cloning process with inputs $P_a$, $P_b$ and target outputs $P_{aa}$, $P_{bb}$. Suppose one observes the operational features O1-O2.
Then, for any noncontextual model we have that
\begin{align}\label{eq:ncbound-ideal}
F_g \leq F^{\rm NC}_g = ~ 1- \frac{c_{ab}}{2} + \frac{c_{aa,bb}}{2}.
\end{align}
\end{thm}
\begin{proof}
The first part of the proof essentially follows the argument given in Ref.~\cite{spekkens2005contextuality} Sec.~VIIIA and reproduced in Ref.~\cite{schmid2018contextual} Sec.~IVA, slightly adapted to use the fewer assumptions of the statement. We have that
\begin{equation*}
1= p(M_k|P_k) = \int_{S_k} d \lambda \mu_k(\lambda) \xi_{k}(\lambda), \quad k=s,s',
\end{equation*}
where $S_k$ denotes the support of $\mu_k$. From this equation, it follows that $\xi_{k}(\lambda) =1$ almost everywhere on $S_k$ (that is, modulo sets of measure zero). Furthermore,
\begin{equation*}
0 = p(M_k|P_{k^\perp}) = \int_{S_{k^\perp}} \mu_{k^\perp}(\lambda) \xi_{k}(\lambda),  \quad k=s,s',
\end{equation*}
from which it follows that $\xi_{k}(\lambda) =0$ almost everywhere on $S_{k^\perp}$. Hence, $S_k \cap S_{k^\perp} = \emptyset$ modulo sets of zero measure.

The operational equivalence of assumption~1 implies that in a noncontextual model 
\begin{equation}
\label{eq:inlemmaprepNC}
\mu_s(\lambda) + \mu_{s^\perp}(\lambda) = \mu_{s'}(\lambda) + \mu_{s'^\perp}(\lambda), \quad \forall \lambda \in \Lambda.
\end{equation}
Since $S_s \cap S_{s^\perp} = S_{s'} \cap S_{s'^\perp} = \emptyset$ modulo a set of zero measure, this implies $ \mu_s(\lambda) = \mu_{s'}(\lambda)$ for almost all $\lambda \in S_s \cap S_{s'}$. Hence, using the facts above, the $\ell_1$ norm distance between $\mu_s$ and $\mu_{s'}$ reads ($\| \mu_s - \mu_{s'}\| := \int d\lambda |\mu_{s}(\lambda) - \mu_{s'}(\lambda)|$).
\begin{align*}
\| \mu_s - \mu_{s'} \| & = \int_{\Lambda \backslash S_s} d \lambda \mu_{s'}(\lambda) + \int_{\Lambda \backslash S_{s'}} d \lambda \mu_{s}(\lambda)  \\ & = 2 - 2 \int_{ S_s \cap S_{s'} } d \lambda \mu_s(\lambda) \\& = 2 - 2 \int_{ S_s \cap S_{s'} } d \lambda \mu_s(\lambda) \xi_{s'}(\lambda).
\end{align*} 
Note that the last integral can be extended to $\Lambda$. In fact, by contradiction suppose that $\xi_{s'}(\lambda) \neq 0$ for some nonzero measure set $X \subseteq S_s \backslash S_{s'}$. Then, from Eq.~\eqref{eq:inlemmaprepNC}, it follows that, for almost all \mbox{$\lambda\in X$}, \mbox{$0 < \mu_s(\lambda) = \mu_{s'^\perp}(\lambda)$}. However, as we discussed $\xi_{s'}(\lambda) = 0$ almost everywhere on $S_{s'^\perp}$, which gives the desired contradiction. Hence the integral can be extended to $S_s \cup S_{s'}$ and, trivially, to all $\Lambda$. In conclusion,
\begin{equation}\label{eq:norm-confusab-ideal}
\| \mu_s - \mu_{s'} \| = 2 - 2 \int_{ \Lambda } d \lambda \mu_s(\lambda) \xi_{s'}(\lambda) = 2(1- c_{ss'}),
\end{equation}
where $c_{ss'} = p(M_{s'}|P_s)$.
Using the triangle inequality,
	\begin{equation*}
	\| \mu_{aa} - \mu_{bb} \| \leq \| \mu_{aa} - \mu_\alpha \| + \| \mu_{\alpha} - \mu_\beta \| + \| \mu_\beta - \mu_{bb}\|.
	\end{equation*}
	By definition, $\mu_{\alpha} (\lambda) = \int d\lambda'  T(\lambda|\lambda') \mu_a (\lambda')$, for a stochastic matrix $T(\lambda|\lambda')$. Similarly, $\mu_{\beta} (\lambda) = \int d\lambda' T(\lambda|\lambda') \mu_b (\lambda')$, with the same stochastic matrix. Since $\int d \lambda T(\lambda|\lambda') =1$ and $ T(\lambda|\lambda') \geq 0$, one can readily verify from the convexity of the absolute value that 
	$\| \mu_\alpha - \mu_\beta \| \leq \| \mu_a - \mu_b \|$ (data processing inequality), which implies
	\begin{equation}
	\label{eq:proofth1}
	\| \mu_{aa} - \mu_{bb} \| \leq \| \mu_{\alpha} - \mu_{aa} \| + \| \mu_a - \mu_b \| + \| \mu_\beta - \mu_{bb}\|.
	\end{equation}
	We can apply Eq.~\eqref{eq:norm-confusab-ideal} to each of the couples $(s,s')$ on the right hand side of Eq.~\eqref{eq:proofth1}, obtaining 
	\begin{equation}
	\label{eq:penultimate}
	 \| \mu_{aa} - \mu_{bb} \|  \leq 2(1-c_{\alpha aa}) + 2(1- c_{ab}) + 2(1- c_{\beta bb}).
	\end{equation}
	Let us now show $ \| \mu_{aa} - \mu_{bb} \| \geq 	2 (1- c_{aa,bb})$. First, notice that
	\begin{align*}
	\| \mu_{aa} - \mu_{bb} \| = \int\limits_{\substack{S_{aa} \backslash S_{bb}}}d\lambda\mu_{aa}(\lambda)+
					\int\limits_{\substack{S_{bb} \backslash S_{aa}}}d\lambda\mu_{bb}(\lambda)~+\\\int\limits_{\substack{R_1}}d\lambda(\mu_{aa}(\lambda)-\mu_{bb}(\lambda))+
	\int\limits_{\substack{R_2}}d\lambda(\mu_{bb}(\lambda)-\mu_{aa}(\lambda)),
	\end{align*}
	with $R_1:=\{\lambda\in S_{aa} \cap S_{bb}:\mu_{aa}(\lambda)\geq\mu_{bb}(\lambda)\}$ and $R_2:=(S_{aa} \cap S_{bb})\backslash R_1$. 
	Next,
	\begin{align*}
		\| \mu_{aa} - \mu_{bb} \|&=2\left(1-\int\limits_{\substack{R_1}}d\lambda\mu_{bb}(\lambda)-\int\limits_{\substack{R_2}}d\lambda\mu_{aa}(\lambda)\right)\\
		&\geq 2 -2\int\limits_{\substack{R_1\cup R_2=S_{aa} \cap S_{bb}}}d\lambda\mu_{aa}(\lambda)	\\
		&= 2-2\int\limits_{\substack{ S_{aa} \cap S_{bb}}}d\lambda\mu_{aa}(\lambda)\xi_{bb}(\lambda)\\
		&\geq 2(1-c_{aa,bb})
\end{align*}		
where the first inequality follows from $\mu_{aa}(\lambda)\geq\mu_{bb}(\lambda)~\forall\lambda\in R_1$ and the second equality follows from $\xi_{bb}(\lambda)=1$ almost everywhere in $S_{bb}$. Finally, substituting this in Eq.~\eqref{eq:penultimate} and rearranging the terms gives
\begin{equation*}
\frac{1}{2} c_{\alpha aa} + \frac{1}{2} c_{\beta bb} \leq 1- \frac{c_{ab}}{2} + \frac{c_{aa,bb}}{2}  
\end{equation*}
and since $F_g = \frac{1}{2} c_{\alpha aa} + \frac{1}{2} c_{\beta bb}$ the global
cloning achieved by non-contextual ontological models that comply with the operational
features O1-O2 is upper bounded as in Eq.~\eqref{eq:ncbound-ideal}.
\end{proof}

In Fig.~\ref{fig:quantumvsclassical} we compare the optimal quantum cloning (global) fidelity of Eq.~\eqref{eq:quantumbound} with the maximum noncontextual cloning fidelity of Eq.~\eqref{eq:ncbound-ideal}, taking into account that, in quantum experiments, one observes $c_{aa,bb}=c^2_{ab}$. One can see, for any \mbox{$0<c_{ab}<1$}, that quantum mechanics achieves higher copying fidelities than what is allowed by the principle of noncontextuality. Hence, the phenomenology of optimal cloning cannot be reproduced within noncontextual ontological models. Contextuality provides an advantage for the maximum copying fidelity.\footnote{Of course, when  $c_{ab} = 0$ - as it is for classical, i.e., orthogonal, states - both the quantum and the noncontextual fidelities are $1$.}

\begin{figure}[t]
\includegraphics[width=0.95\columnwidth]{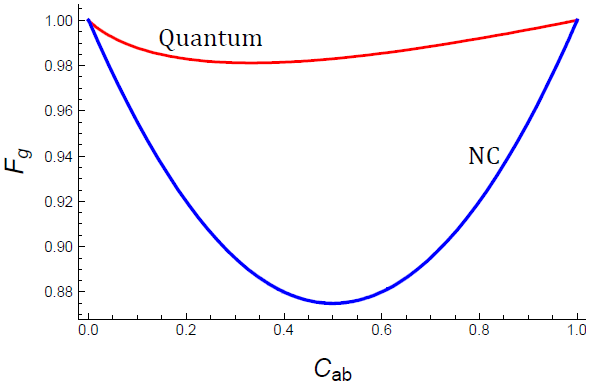}
\caption{Maximum tradeoff between cloning fidelity $F_g$ and confusability $c_{ab}$ allowed for noncontextual models (blue line, Eq.~\eqref{eq:ncbound-ideal}) versus optimal tradeoff achievable in quantum theory (red line, Eq.~\eqref{eq:quantumbound}).
\label{fig:quantumvsclassical}}
\end{figure}

Interestingly, the above derivation also gives an alternative, simple proof of the main result of Ref.~\cite{schmid2018contextual}. In fact, an intermediate technical result in the proof of Theorem~\ref{thm:cloningcontextual-ideal} is
that in the presence of the operational features O1-O2,
noncontextual models must have a direct relation between the experimentally accessible confusabilities $c_{ss'} = p(M_{s'}|P_s)$ and the $\ell_1$ distance between the corresponding probability distributions: 
\begin{equation}
\label{eq:confusabilitytracenorm-ideal}
\| \mu_s - \mu_{s'} \|  = 2(1- c_{ss'}).
\end{equation}
(This was implicitly shown in Ref.~\cite{schmid2018contextual} Sec. IVA, but using infinitely many extra operational assumptions. That is, they assume O2 for all pairs of orthogonal states).

Since the maximum probability $s_{ab}$ of distinguishing two preparations $P_a$ and $P_{b}$ is at most $1/2 +$ $\| \mu_a - \mu_{b}\|/4$,
it immediately follows $s_{ab}\leq 1- c_{ab}/2$, which is the optimal state discrimination probability in noncontextual models, as given in Ref.~\cite{schmid2018contextual}. Conversely, it is not immediately obvious how the techniques of Ref.~\cite{schmid2018contextual} could be adapted to obtain our result on cloning, due to our use of the data processing inequality in Theorem~\ref{thm:cloningcontextual-ideal}.

We also note that the noncontextual bound on cloning is tight. Denote by $S_s$ the support of $\mu_s$.  Consider a model in which $\mu_{ss} = \mu_s \mu_{s}$ and $\xi_s(\lambda)=1$ if $\lambda \in S_s$ and zero otherwise.  A cloning strategy that saturates the bound is as follows: if the input $\lambda$ is in $S_a \backslash S_b$, output $(\lambda,\lambda')$, with $\lambda'$ sampled according to $\mu_a$; otherwise, output $(\lambda, \lambda')$ with $\lambda'$ sampled according to $\mu_b$. Notice that this sets $\mu_{\beta} = \mu_b \mu_b$ and, hence, $c_{\beta bb}=1$ ($\mu_b$ is copied perfectly). On the other hand, $\mu_{\alpha}(\lambda,\lambda') = \mu_a(\lambda) \mu_a (\lambda')$ for  $\lambda \in S_a \backslash S_b$ and $\mu_{\alpha}(\lambda,\lambda') = \mu_a(\lambda) \mu_b (\lambda')$ for  \mbox{$\lambda \in S_a \cap S_b$} and, hence, 
\begin{align*}
c_{\alpha aa}&=\int d\lambda d\lambda' \mu_{\alpha}(\lambda,\lambda') \xi_{aa}(\lambda,\lambda')\\
&=\int_{S_a\times S_a} d\lambda d\lambda' \mu_{\alpha}(\lambda,\lambda')\\
&=\int_{(S_a\setminus S_b)\times S_a}d\lambda d\lambda'\mu_{a}(\lambda)\mu_{a}(\lambda')\quad+\\
&\qquad\qquad\int_{(S_a\cap S_b)\times S_a}d\lambda d\lambda'\mu_{a}(\lambda)\mu_{b}(\lambda')\\
&=(1-c_{ab})+c_{ab}\cdot c_{ba}=1-c_{ab}+c_{ab}^2,
\end{align*}
where, in the last equality, we use the operational fact that $c_{ab}=c_{ba}$. Finally, this gives \mbox{$F_g = \frac{1}{2}(1-c_{ab} + c_{ab}^2) + \frac{1}{2} = F^{\rm NC}_g$}. In Appendix~\ref{appendix:nc-strategy} we complete this strategy with a concrete choice of $\mu_a$, $\mu_b$, $\mu_{{aa}^\perp}$, $\mu_{{bb}^\perp}$, $\mu_{{\alpha}^\perp}$ and $\mu_{{\beta}^\perp}$ complying with O1 and satisfying Eq. \eqref{eq:noncontextualityconsequence} for all the operational equivalences in O2.

This optimal strategy seems to suggest the following intuition behind the theorem: our assumption of preparation noncontextuality on the input preparations $P_a$, $P_b$ imply that the distributions $\mu_a(\lambda)$ and $\mu_b(\lambda)$ overlap ``too much'' (formally, it implies maximal $\psi$-epistemicity, $c_{ab} = \int_{S_b} d\lambda \mu_a(\lambda)$ \cite{leifer2013maximally}), hence the cloning performance turns out worse than in quantum mechanics. Furthermore, noncontextuality implies that $\mu_a$ and $\mu_b$ coincide on their overlap, which implies a direct relation between $c_{ab}$ and the $\ell_1$ norm $\| \mu_a - \mu_b\|$. Crucially the latter cannot be increased by the cloning machine, since $\| \cdot\|$ decreases under post-processing.

However, this mechanism can only be part of the story. First, the cloning performance is not monotonically decreasing with increasing overlap, since for $c_{ab}=1$ one can clone perfectly. Second, cloning is defined as the creation of two independent copies of the preparations $P_a$ or $P_b$, but these do not necessarily correspond to two independent copies $\mu_a \mu_a$, $\mu_b \mu_b$ (this assumption, which we do not make, is called \emph{preparation independence} \cite{pusey2012reality}). Nevertheless, we showed that a no-go theorem results from the observed overlaps $c_{ab}$, $c_{aa,bb}$ and noncontextuality assumptions only as a consequence of information processing inequalities and the triangle inequality. 
	
	We note in passing that our proof technique can be abstracted and applied to other tasks as follows:
\begin{enumerate}
	\item First, given a set of observed overlaps $\{c_{ss'}\}$, noncontextuality applied to the operational equivalences $\frac{1}{2}P_s + \frac{1}{2}P_{s^\perp} \simeq \frac{1}{2} P_{s'} + \frac{1}{2}P_{s^{'\perp}}$ gives the equations~\eqref{eq:confusabilitytracenorm-ideal}.
\item Second, verify if the equations~\eqref{eq:confusabilitytracenorm-ideal} are compatible with triangle and data processing inequalities and the performance of quantum protocol under consideration (in this case, state-dependent cloning).
\end{enumerate}
In fact, the same proof technique can be extended to nonideal scenarios (with Eq.~\eqref{eq:confusabilitytracenorm-ideal} replaced by Eq.~\eqref{eq:trace-norm-confus-noisy}), as we now see.

\section{Optimal cloning is contextual - beyond idealizations}

Theorem~\ref{thm:cloningcontextual-ideal} is a no-go result for noncontextual ontological models aimed at explaining the phenomenology of state-dependent quantum cloning. However, the inequality derived in Eq. \eqref{eq:ncbound-ideal} is not a proper noncontextuality inequality because the operational features considered refer to an idealized experiment. In any real experiment, on the other hand, one will need to confront the following nonidealities: 
\begin{itemize}
\item The correlations in O1 will only approximatively hold in data collected in a real experiment.
\item O2 will only be approximatively realized.
\end{itemize}
Theorem~\ref{thm:cloningcontextual-noisy} below extends Theorem~\ref{thm:cloningcontextual-ideal} beyond the ideal limit, allowing for the observation of nonperfect correlations in O1, such as those generated by a cloning experiment carried out with nonideal preparations and test measurements. As we will discuss later, there are general techniques to deal with the idealization in O2, so that the problem of deriving an experimentally testable statement reduces to the elimitation of the idealization in O1. Specifically, we want to weaken it to
\begin{enumerate}
	 \item[O1ni] $p(M_s|P_s) \geq 1-\epsilon_s, \quad p(M_s|P_{s^\perp}) \leq  \epsilon_s$  for $s = a,b,\alpha,\beta,aa,bb$,
\end{enumerate}
where `ni' stands for `non-ideal'. 

\begin{thm}[Optimal cloning fidelity in noncontextual models -- noise-robust version]
\label{thm:cloningcontextual-noisy}
With the notation of Thm. \ref{thm:cloningcontextual-ideal}, suppose that one observes the operational features O1ni and O2. 
 Then, for any noncontextual model we have that
\begin{align}\label{eq:ncbound-noisy}
F_g \leq F^{\rm NC, ni}_g = 1- \frac{c_{ab}}{2} + \frac{c_{aa,bb}}{2}+ \emph{Err}.
\end{align}
where $\emph{Err} = \frac{1}{2}({\epsilon_{b}+2\epsilon_{bb}+\epsilon_{aa}})$.
\end{thm}
Note that, while we gave an independent and simpler proof of Theorem~\ref{thm:cloningcontextual-ideal}, we can now see it as a corollary of the result above once all error terms are set of zero. Another interesting case is when all error terms are equal, $\epsilon_b = \epsilon_{bb} = \epsilon_{aa} := \epsilon$, which gives 
$
F^{\rm NC,ni}_g = 1- \frac{c_{ab}}{2} + \frac{c_{aa,bb}}{2} + 2 \epsilon$.  In fact, we can give a slightly stronger and symmetric bound than the above. For the specific form, see Appendix~\ref{appendix:proof-noisy-thm}.

The proof of Theorem~\ref{thm:cloningcontextual-noisy} follows the same lines as that of Theorem~\ref{thm:cloningcontextual-ideal}. The key addition is to extend Eq.~\eqref{eq:confusabilitytracenorm-ideal} to the noisy setting. Specifically, we show that in the presence of the operational features O1ni-O2, noncontextual models must satisfy
\begin{equation}
\label{eq:trace-norm-confus-noisy}
		 |\|\mu_s - \mu_{s'}\|-2 (1- c_{ss'})| \leq 2 \epsilon_{s'},
	\end{equation}
and similarly if we exchange $s$ and $s'$.
In other words, the relation of Eq.~\eqref{eq:confusabilitytracenorm-ideal} holds approximatively, and we can bound its violation with the experimentally accessible noise level. The proof of this result is more involved than in the ideal scenario, so we postpone the derivation to Appendix~\ref{appendix:proof-noisy-thm}.
	
Eq.~\eqref{eq:trace-norm-confus-noisy} imposes a strict relation, in any noncontextual model and beyond the ideal scenario, between the $\ell_1$ distance of two epistemic states and their operationally accessible confusability. Hence, we anticipate that these relations will be of broader use to identify quantum advantages beyond state-dependent cloning. For instance, following the same reasoning given after Theorem~\ref{thm:cloningcontextual-ideal}, these inequalities provide an alternative and intuitive derivation of the tight noise-robust noncontextual bound on state discrimination of Ref.~\cite{schmid2018contextual},
$s_{ab}\leq \frac{1}{2} + \frac{1}{4} \|\mu_a - \mu_b\| \leq 1 - \frac{c_{ab}-\epsilon_b}{2}$.

\subsection{An explicit noise model}
 Having derived a noise-robust version of our noncontextual bound, the next step is to investigate whether quantum mechanics violates it. We consider a standard noise model in which the ideal quantum preparations, measurements and unitary transformation are all thwarted by a depolarizing channel $\mathcal{N}_v$ with noise level $v\in [0,1]$:
\begin{align*}
\mathcal{N}_v(\rho)=(1-v)~\rho+v \mathbb{I}/{4}.
\end{align*}

A direct calculation (see Appendix~\ref{appendix:quantum-model}) shows that this sets $\epsilon = v (31 - 21 v + 9 v^2)/16$ in Eq.~\eqref{eq:ncbound-noisy}. If one uses the unitary transformation that is optimal for state-dependent cloning in the noiseless setting, one gets a quantum strategy whose global average fidelity reads 
\begin{equation}\label{eq:quantum-value-noisy}
F_g^{\rm Q,noisy}(v):=(1-v)^3F_g^{\rm Q,opt}+\frac{1}{4}v(3-3v+v^2)
\end{equation}
which coincides with the optimal for $v=0$. For $v>0$, however, and unlike in the ideal case, the tradeoff between $c_{ab}$ and $F_g$ is not necessarily above the noncontextual bound. For example, for $v=0.015$ a violation can be observed only for $c_{ab} \in [0.318,0.718]$, see Fig.~\ref{fig:tradeoff-lambda-vs-confusability}. Nevertheless, a preliminary comparison with the experimental results of Ref.~\cite{mazurek16} suggests that the required low level of noise is not beyond current experiments. In fact, in terms of the parameter $C_s = 1/2\, p(M_s|P_s) + 1/2\, p(M_{s^\perp}|P_{s^\perp})$ defined in Ref.~\cite{mazurek16} ($C_s= 1$ in the ideal scenario), $v=0.015$ corresponds to $C_s \approx 0.9851$ for $s=a,b$ and $C_s \approx 0.9667$ for $s=aa,bb$, and Ref.~\cite{mazurek16} experimentally realized $C_s=0.9969$.

\begin{figure}[t]
	\centering
	\includegraphics[width=0.95\linewidth, height=5cm]{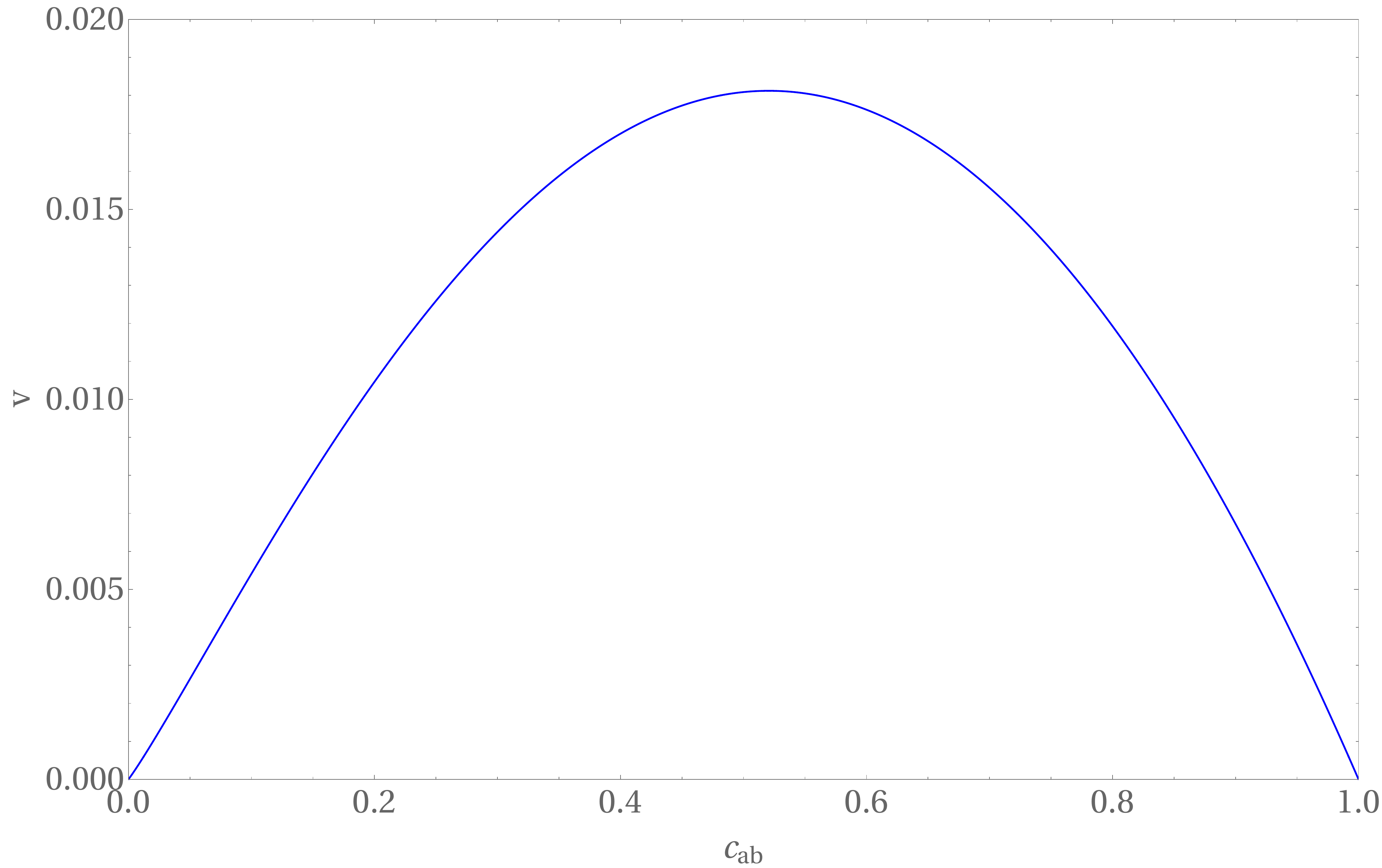}
	\caption{\emph{Noise-resistance of the quantum advantage in cloning.} This plot shows the maximum value of the noise parameter $v$ of a depolarizing channel (affecting preparations, measurements and transformation) for which the quantum value of the cloning fidelity (Eq.~\eqref{eq:quantum-value-noisy}) is above the noncontextual bound, as a function of the confusability between the inputs $c_{ab}$. 
	\label{fig:tradeoff-lambda-vs-confusability}}
\end{figure}

\subsection{Remaining assumptions}
As we mentioned, the only remaining idealization is the operational verification of O2. Let us suppose that, in an experiment, after doing tomography,\footnote{We will later discuss the assumption that can one access a tomographically complete set of measurements.} one determines that the actual experimental realizations of the ideal preparations are 
$P^{(1)}_a,P^{(1)}_{a^\perp},P^{(1)}_b,P^{(1)}_{b^\perp},P^{(1)}_{aa},P^{(1)}_{{aa}^\perp},P^{(1)}_{bb},P^{(1)}_{{bb}^\perp},P^{(1)}_{\alpha},P^{(1)}_{{\alpha}^\perp}$, $P^{(1)}_{\beta}$ and $P^{(1)}_{{\beta}^\perp}$. These 'primary' preparations will, in general, not respect the required operational equivalences in O2, due to unavoidable imperfections in the experimental realisation. Luckily, there are general considerations to tackle this idealization~\cite{mazurek16}. 

The first thing to notice is that if one can experimentally achieve a set of preparations  $P^{(1)}_s$, then one can also prepare any convex combination of them, i.e. any preparation in the convex hull $\mathcal{C}$ of the preparations $P^{(1)}_s$. By the linearity of Eq.~\eqref{eq:probability-rule}, one can then compute the measurement statistics of all the preparations in $\mathcal{C}$. Therefore, as put forward in Sec.~IV of Ref.~\cite{mazurek16}, to go ahead with the experimental verification of Theorem \ref{thm:cloningcontextual-noisy} one only needs to find `secondary' preparations $P^{(2)}_s$ in $\mathcal{C}$ whose measurement statistics satisfy the operational equivalences in O2; that is, we only need

\begin{enumerate}
	 \item[O2ni] O2 is satisfied for some preparations $P^{(2)}_s$ in the convex hull $\mathcal{C}$ of the experimental preparations $P^{(1)}_s$.
\end{enumerate}
This post-processing hence allows one to apply Theorem~\ref{thm:cloningcontextual-noisy} even if the collected data does not satisfy O2. One can think of the secondary preparations as noisy versions of the primary preparations. Hence, the price one pays in this construction is that the corresponding noise parameters $\epsilon'_s = p(M_s|P^{(2)}_s)$ in O1ni will in general be larger. Note that, even if $\epsilon'_s$ is too large compared to $\epsilon_s$ to see any violation in Theorem~\ref{thm:cloningcontextual-noisy}, one can get around this issue by adding extra experimental preparations $P^{(1)}_{\rm extra}$ to enlarge $\mathcal{C}$, as explicitly done in Ref.~\cite{mazurek16}. To summarize, there are good general tools to deal with imperfections in the operational equivalences O2.

 As a final remark, it is useful to briefly talk about loopholes. These are all those assumptions that cannot be \emph{conclusively} tested by any experimental means. In a nonlocality experiment, for example, these include the assumption that the two sides cannot communicate and the ability to choose the measurement freely, i.e., independently of any other variable relevant to the experiment. In a contextuality experiment the notion of operational equivalence relies on the knowledge of a tomographically complete set of measurements. However, if quantum theory is not correct, the tomographically complete set of a post-quantum theory may contain extra unknown measurements (just like a future theory may allow signalling). Recent work has shown that the problem can be mitigated by the addition of extra (known) measurements and preparations (see Ref.~\cite{pusey2019contextuality}), but this goes beyond the scope of the present work. 

\section{Conclusions and open questions.}
We have shown that the operational statistics observed in the optimal state-dependent quantum cloning is incompatible with the predictions of every noncontextual ontological model. In particular, for given overlap, the noncontextual global cloning fidelity is strictly smaller than the quantum prediction. A similar result continues to hold in more realistic experiments which are unavoidably affected by noise (while the effect can be `washed out' by excessive experimental imperfections). This identifies contextuality as the resource for optimal  state-dependent quantum cloning.

From a foundational point of view, it would be relevant to explore whether the relation between contextuality and cloning fidelity, 
that we proved for optimal state-dependent cloning, extends to the other types of imperfect cloning studied in the literature, mainly phase-covariant and/or universal cloning, as well as to probabilistic cloning \cite{scarani2005cloning}. From an applications' point of view, one important open question is if our noncontextual bound can be used to prove a contextual advantage for quantum information processing tasks which rely on optimal quantum state-dependent cloning (\emph{e.g.}, \cite{scarani2009security,deuar2000quantum}).

Finally, it may be possible to use the connection between $\ell_1$ norm and confusability developed here to understand what aspects of other quantum information primitives, such as quantum teleportation, are truly nonclassical. 

\bigskip

\emph{Acknowledgements.} We are grateful to Joseph Bowles for useful comments on a draft of this manuscript. We acknowledge financial support from the the European Union's Marie Sklodowska-Curie individual Fellowships (H2020-MSCA-IF-2017, GA794842), Spanish MINECO (Severo Ochoa SEV-2015-0522 and project QIBEQI FIS2016-80773-P), Fundacio Cellex and Generalitat de Catalunya (CERCA Programme and SGR 875).

\bibliographystyle{unsrtnat}
\bibliography{biblio}

\onecolumn\newpage 

\appendix

\section{Noncontextual model saturating the bound in Theorem \ref{thm:cloningcontextual-ideal}}\label{appendix:nc-strategy}
To complement the cloning strategy given in the main text, in this section we give a concrete choice of distributions $\mu_{aa}$, $\mu_{bb}$, $\mu_{{aa}^\perp}$, $\mu_{{bb}^\perp}$ and $\mu_{{\alpha}^\perp}$ satisfying the operational features targeted by Theorem \ref{thm:cloningcontextual-ideal}.  The supports of all these distributions, which we set to be subsets of $[0,2]\times[0,2]$, are plotted in Figure \ref{fig:nc-distributions}. All the distributions are constantly $1$ on their support. Notice that since the cloning map given in the main text makes $\mu_\beta\equiv\mu_{bb}$, it follows that to satisfy the operational equivalence for $(\mu_{\beta},\mu_{bb})$ we must have $\mu_{{\beta}^\perp}\equiv \mu_{{bb}^\perp}$. We let the reader verify, by inspecting the plots, that the remaining requirements implied by the operational features O1 and O2 are satisfied by these distributions (and the choice of response functions made in the main text).

\begin{figure}[h!]
\centering
\begin{subfigure}{.5\textwidth}
  \centering
  \includegraphics[width=.65\linewidth]{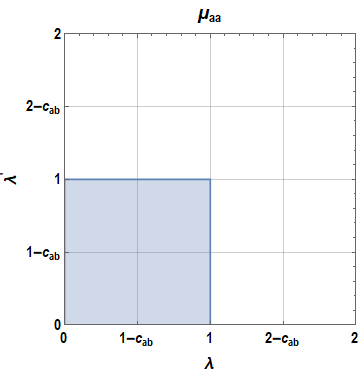}
  \caption{$S_{aa}$}
  \label{fig:sub1}
\end{subfigure}%
\begin{subfigure}{.5\textwidth}
  \centering
  \includegraphics[width=.65\linewidth]{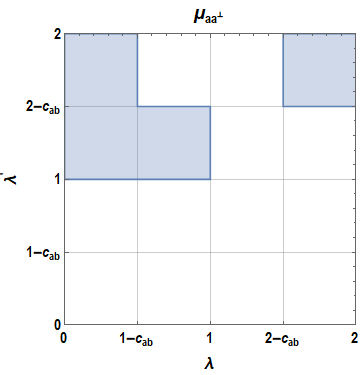}
  \caption{$S_{{aa}^\perp}$}
  \label{fig:sub2}
\end{subfigure}
\begin{subfigure}{.5\textwidth}
  \centering
  \includegraphics[width=.65\linewidth]{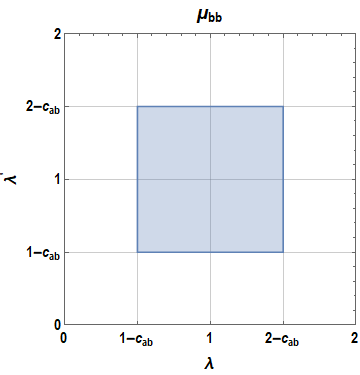}
  \caption{$S_{{bb}}$}
  \label{fig:sub3}
\end{subfigure}%
\begin{subfigure}{.5\textwidth}
  \centering
  \includegraphics[width=.65\linewidth]{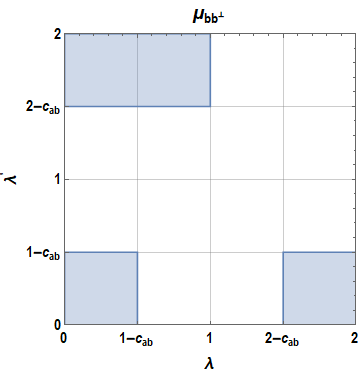}
  \caption{$S_{{bb}^\perp}$}
  \label{fig:sub4}
\end{subfigure}
  \begin{subfigure}{.5\textwidth}
  \centering
  \includegraphics[width=.65\linewidth]{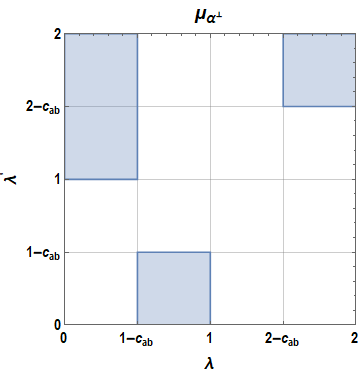}
  \caption{$S_{{\alpha}^\perp}$}
  \label{fig:sub3}
\end{subfigure}
\caption{Supports of distributions $\mu_{aa}$, $\mu_{bb}$, $\mu_{{aa}^\perp}$, $\mu_{{bb}^\perp}$ and $\mu_{{\alpha}^\perp}$ satisfying the restrictions imposed on noncotextual models by the requirements of Theorem \ref{thm:cloningcontextual-ideal}. The distributions are $1$-valued in the filled regions (i.e. in their support).}
\label{fig:nc-distributions}
\end{figure}

\section{Generalization and proof of Theorem~\ref{thm:cloningcontextual-noisy}} \label{appendix:proof-noisy-thm}
In this section we will prove a slightly stronger and more symmetric bound on the noncontextual cloning fidelity $F^{\rm NC}_g$ from which the bound in Thm. \ref{thm:cloningcontextual-noisy} in the main text follows straightforwardly as a corollary.
 
\begin{thm}
\label{thm:cloningcontextual-noisy-full}
With the notation of Thm. \ref{thm:cloningcontextual-ideal}, suppose that one observes the operational features O1ni and O2:
\begin{enumerate}
	\item[O1ni] $p(M_s|P_s) \geq 1-\epsilon_s, \quad p(M_s|P_{s^\perp}) \leq  \epsilon_s$  for $s = a,b,\alpha,\beta,aa,bb,$
	\item[O2] $\frac{1}{2} P_s + \frac{1}{2} P_{s^\perp} \simeq \frac{1}{2} P_{s'}+ \frac{1}{2} P_{s^{'\perp}}$, for all $(s,s')$ in $\{(a,b), (\alpha, aa), (\beta, bb) \}$.
\end{enumerate}
 Then, for any noncontextual model we have that
\begin{align}\label{eq:ncbound-noisy-appendix}
F^{\rm NC}_g \leq 1 +\frac{\min\{\epsilon_b-c_{ab},\epsilon_a  -c_{ba
}\}}{2}+\frac{\min\{c_{aa,bb}+\epsilon_{bb},c_{bb,aa}+\epsilon_{aa}\}}{2}+ \frac{\epsilon_{aa}+\epsilon_{bb}}{2}.
\end{align}
\end{thm}

For the proof of Theorem~\ref{thm:cloningcontextual-noisy-full}, we make use of the following lemma relating the $\ell_1$ distance of two epistemic states in any ontological model satisfying the hypothesis of the theorem and their operationally accessible confusability.

\begin{lem}\label{lem:noisy-symm}
	Let $P_s$, $P_{s'}$ be preparations. Suppose there exists preparations $P_{s^\perp}$, $P_{s'^\perp}$ and a two outcome measurement $M_s$ such that
	\begin{enumerate}
		\item $\frac{1}{2}P_s + \frac{1}{2}P_{s^\perp} \simeq  \frac{1}{2}P_{s'} + \frac{1}{2}P_{s'^\perp},$  \label{ass:opequivbis}		
		\item $p(M_k|P_k) \geq 1-\epsilon_k$, \quad $p(M_k|P_{k^\perp}) \leq \epsilon_k$, \quad $k=s,s'.$\label{ass:puritybis}
	\end{enumerate} 
	Then, in a noncontextual ontological model,
	\begin{equation}
		2 \max\{1- c_{ss'}-\epsilon_{s'},1- c_{s's}-\epsilon_{s}\} \leq \|\mu_s - \mu_{s'}\| \leq 2 \min\{1- c_{ss'}+\epsilon_{s'},1- c_{s's}+\epsilon_{s}\},
	\end{equation}\label{eq:trace-norm-confus-noisy-symm}
\end{lem}
\begin{proof}
We denote by $S_s$ the support of $\mu_s$. Define a partition $S_s \cup S_{s'} = \sqcup_{i=1}^4 R_i $, as summarized in Figure~\ref{fig:regions}:
	\begin{itemize}
	\item $R_1 = S_s \backslash (S_s \cap S_{s'})$, $R_{4} = S_{s'} \backslash (S_s \cap S_{s'})$.
	\item $R_2 = \{\lambda \in S_s \cap S_{s'} | \mu_s(\lambda) \geq \mu_{s'}(\lambda) \}$, $R_3 = \{\lambda \in S_s \cap S_{s'} | \mu_s(\lambda) < \mu_{s'}(\lambda) \}$.
	\end{itemize}

\begin{figure}[h!]
\begin{centering}
	\includegraphics[width=0.35\columnwidth]{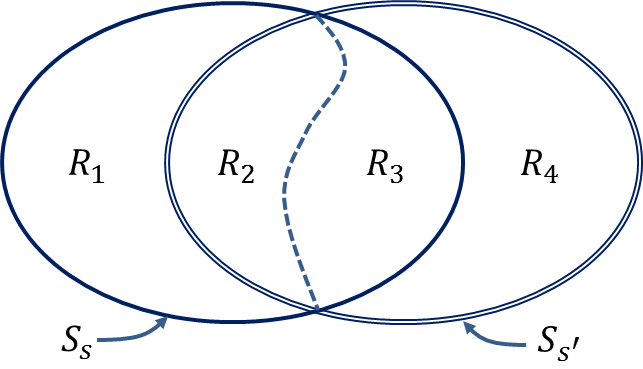}
	\caption{Sketch of the relevant regions in the proof of Lemma~\ref{lem:noisy-symm}.}
		\label{fig:regions}
\end{centering}
\end{figure}
Then,
	\begin{align}
	\| \mu_s - \mu_{s'} \| &=  \int d\lambda |\mu_s(\lambda) - \mu_{s'}(\lambda)| \nonumber \\
	&= \int_{R_1} d\lambda \mu_s(\lambda) + \int_{R_4} d\lambda \mu_{s'}(\lambda) + \int_{R_2} d\lambda [\mu_s(\lambda) - \mu_{s'}(\lambda)] + \int_{R_3} d\lambda [\mu_{s'}(\lambda) - \mu_{s}(\lambda)] \nonumber \\
	  &=  2- \int_{R_2 \cup R_3} d\lambda[\mu_{s}(\lambda) + \mu_{s'}(\lambda)] + \int_{R_2} d\lambda [\mu_s(\lambda) - \mu_{s'}(\lambda)] + \int_{R_3} d\lambda [\mu_{s'}(\lambda) - \mu_{s}(\lambda)] \nonumber \\
	  &=  2- 2  \int_{R_3} d\lambda \mu_s(\lambda) - 2 \int_{R_2} d\lambda \mu_{s'}(\lambda).
	  \label{eq:lemmanoisysym}
	\end{align}
Consider,
\begin{align*}
c_{ss'} - \int_{R_3} d\lambda \mu_s(\lambda) - \int_{R_2} d\lambda \mu_{s'}(\lambda) &= \int_{R_1\cup R_2\cup R_3} d\lambda \mu_{s} (\lambda) \xi_{s'}(\lambda)- \int_{R_3} d\lambda \mu_s(\lambda) - \int_{R_2} d\lambda \mu_{s'}(\lambda)\\
 &\leq \int_{R_1\cup R_2} d\lambda \mu_{s} (\lambda) \xi_{s'}(\lambda) - \int_{R_2} d\lambda \mu_{s'}(\lambda)\\
&= \int_{R_1\cup R_2} d\lambda [\mu_{s'}(\lambda)+\mu_{s^{'\bot}}(\lambda)-\mu_{s^{\bot}}(\lambda)]\xi_{s'}(\lambda)- \int_{R_2} d\lambda \mu_{s'}(\lambda)\\
&=- \int_{R_2} d\lambda \mu_{s'}(\lambda) \xi_{s^{'\perp}}(\lambda)  + \int_{R_1 \cup R_2} d\lambda  [\mu_{s^{'\perp}}(\lambda) - \mu_{s^\perp}(\lambda)]\xi_{s'}(\lambda) \\
&\leq \int_{R_1 \cup R_2}d\lambda  \mu_{s^{'\perp}}(\lambda) \xi_{s'}(\lambda) \leq \int d\lambda \mu_{s^{'\perp}}(\lambda) \xi_{s'}(\lambda) = p(M_{s'}|P_{s^{'\perp}}) \leq \epsilon_{s'},
\end{align*}
where we used $\xi_{s'} \leq 1$ in the first inequality and assumption~\ref{ass:opequivbis} and non-contextuality in the second equality. In the third equality, we used $\int_{R_1} d\lambda \mu_{s'}(\lambda) \xi_{s'}(\lambda)=0$ and in the final inequality we used assumption~\ref{ass:puritybis}.
Then, using Eq.~\eqref{eq:lemmanoisysym},
\begin{align*}
\| \mu_s - \mu_{s'} \| \leq 2(1-c_{ss'}+\epsilon_{s'}).
\end{align*}
Furthermore, recalling that $\xi_{s^{'\perp}} = 1- \xi_{s'}$,
\begin{align*}
\int_{R_3} d\lambda \mu_s(\lambda)+\int_{R_2} d\lambda \mu_{s'}(\lambda)-c_{ss'} & = 	\int_{R_3} d\lambda \mu_s(\lambda)+\int_{R_2} d\lambda \mu_{s'}(\lambda)-\int_{R_1\cup R_2\cup R_3} d\lambda\mu_{s}(\lambda)\xi_{s'}(\lambda) \\
& \leq \int_{R_3} d\lambda \mu_{s'}(\lambda) \xi_{s{'^\perp}}(\lambda) + \int_{R_2}d\lambda \mu_{s'}(\lambda)\xi_{s{'^\perp}}(\lambda) - \int_{R_1} d\lambda \mu_{s} \xi_{s'}(\lambda) \\
& \leq \int_{R_3} d\lambda \mu_{s'}(\lambda) \xi_{s{'^\perp}}(\lambda) + \int_{R_2}d\lambda \mu_{s'}(\lambda)\xi_{s{'^\perp}}(\lambda)\\
&\leq \int d\lambda \mu_{s'}(\lambda) \xi_{s^{'\perp}}(\lambda) = p(M_{s^{'\perp}}|P_{s'}) \leq \epsilon_{s'}.
\end{align*}
where in the first inequality we used that $\mu_{s}(\lambda) \leq \mu_{s'}(\lambda)$ in $R_3$ and $\mu_{s}(\lambda) \geq \mu_{s'}(\lambda)$ in $R_2$. In the final inequality, we used assumption~\ref{ass:puritybis}. Hence, we have that
\begin{align*}
c_{ss'}-\int_{R_3} d\lambda \mu_s(\lambda)-\int_{R_2} d\lambda \mu_{s'}(\lambda)\geq -\epsilon_{s'}
\end{align*}
and, using Eq.~\eqref{eq:lemmanoisysym}, that
\begin{align}\label{eq:lowerbound-onenorm-noisy}
\| \mu_s - \mu_{s'} \| \geq 2(1-c_{ss'}+\epsilon_{s'}).
\end{align}

Finally, noting that $\|\mu_s - \mu_{s'}\|=\|\mu_{s'} - \mu_s\|$ and that the above derivation is symmetric under the exchange of $s$ with $s'$ we arrive to the desired result
\begin{align}\label{eq:lowerbound-onenorm-noisy-full}
		2 \max\{1- c_{ss'}-\epsilon_{s'},1- c_{s's}-\epsilon_{s}\} \leq \|\mu_s - \mu_{s'}\| \leq 2 \min\{1- c_{ss'}+\epsilon_{s'},1- c_{s's}+\epsilon_{s}\},
\end{align}
Notice that for the lower bound in Eq.~\eqref{eq:lowerbound-onenorm-noisy} (and, hence, the left hand side of Eq. \eqref{eq:lowerbound-onenorm-noisy-full}) we did not use assumption \ref{ass:opequivbis} of operational equivalence.

\end{proof}

Given the above we can now prove Theorem~\ref{thm:cloningcontextual-noisy-full}:
\begin{proof}[Proof of Theorem~\ref{thm:cloningcontextual-noisy-full}]
In the first part we proceed as in the ideal case. From the triangle inequality and the contractivity of the $\ell_1$ norm under stochastic processes (which gives $\|\mu_{\alpha} - \mu_{\beta}\| \leq \|\mu_a - \mu_b\|$), one can show that the following equation holds (see Eq.~\eqref{eq:proofth1}):
\begin{equation}
	\| \mu_{aa} - \mu_{bb} \| \leq \| \mu_{\alpha} - \mu_{aa} \| + \| \mu_a - \mu_b \| + \| \mu_\beta - \mu_{bb}\|.
	\end{equation}
	Using both upper and lower bounds for the $\ell_1$ distance derived in Lemma~\ref{lem:noisy-symm}, this implies
	\begin{align*}
	2 \max\{1- c_{aa,bb}-\epsilon_{bb},1- c_{bb,aa}-\epsilon_{aa}\} \\  \leq 2(1-c_{\alpha aa}) + 2\epsilon_{aa} + 2\min\{1-c_{ab}+\epsilon_b,1-c_{ba}+\epsilon_a\} \\ + 2(1-c_{\beta bb}) + 2\epsilon_{bb},
		\end{align*}
	which can be rearranged to give the claimed bound on $F^{\rm NC}_g$.
\end{proof}

\section{Quantum violation of noise-contextual bound under depolarizing noise}\label{appendix:quantum-model}
\subsection{Introducing noise}
We will assume that all experimental procedures in the ideal quantum cloning experiment (that is, preparations, measurements and transformations) are affected by a depolarizing channel $\mathcal{N}_v$ with noise level $v\in [0,1]$:
\begin{align*}
\mathcal{N}_v(\rho)=(1-v)~\rho+v\frac{\mathbb{I}}{4}.
\end{align*}
Therefore, for $x\in\{a,b\}$, the ideal input preparations transform as
$$\ket{x0}\mapsto \NN_v(\ketbra{x0}{x0})=(1-v)\ketbra{x0}{x0}+v\frac{\Id}{4},$$
so that the actual input preparations become
$$\rho_x:=\Tr_2 \left[\NN_v(\ketbra{x0}{x0})\right]=(1-v)\ketbra{x}{x}+v\frac{\Id_{ab}}{2},$$
with $\Id_{ab}$ the projector over the ${\rm span}(\{\ket{a},\ket{b}\})$. 
The ideal cloning transformation $\UU$ becomes $\NN_v\circ \UU = (1-v)~\UU+v\chn$, where $\chn(\rho)=\mathbb{I}/4$ for all $\rho$. Hence, the actual outcomes $\chi\in\{\alpha,\beta\}$ correspondent to input $x \in \{a,b\}$ become
\begin{align*}
\rho_\chi&:=\mathcal{N}_v\circ \UU (\rho_x) = \left[(1-v)\UU + v\chn\right]\left((1-v)\ketbra{x0}{x0}+v\frac{\Id}{4}\right) 
=(1-v)^2\ketbra{\chi}{\chi}+(1-(1-v)^2)\frac{\Id}{4}.
\end{align*}

The actual target copies would become $\NN_v(\ketbra{xx}{xx})$. While this is the minimal amount of noise in this preparation required by our model, not all operational equivalences are satisfied under it. A simple (albeit likely not optimal) way to fix this issue is to let the noise act for a second step; hence, define
$$\rho_{xx}:= \NN_v \circ \NN_v(\ketbra{xx}{xx})=(1-v)^2\ketbra{xx}{xx}+(1-(1-v)^2)\frac{\Id}{4}.$$

Finally, for the ideal measurements, they transform as (for $x \in \{a,b\}$, $\chi \in   \{\alpha,\beta\}$),
\begin{align*}
\{\ketbra{x}{x},\Id_{ab}-\ketbra{x}{x}\}&\mapsto & M_{x} := &\left\{(1-v)\ketbra{x}{x}+v\frac{\Id_{ab}}{2},(1-v)(\Id_{ab}-\ketbra{x}{x})+v\frac{\Id_{ab}}{2}\right\},\\
\{\ketbra{xx}{xx},\Id-\ketbra{xx}{xx}\}&\mapsto & M_{xx} := &\left\{(1-v)\ketbra{xx}{xx}+v\frac{\Id}{4},(1-v)(\Id-\ketbra{xx}{xx})+v\frac{\Id}{4}\right\}, \\
\{\ketbra{\chi}{\chi},\Id-\ketbra{\chi}{\chi}\}&\mapsto & M_{\chi} := &\left\{(1-v)\ketbra{\chi}{\chi}+v\frac{\Id}{4},(1-v)(\Id -\ketbra{\chi}{\chi})+v\frac{\Id}{4}\right\}.
\end{align*}

\subsection{Orthogonal preparations and operational equivalences}
We now introduce the orthogonal preparations, necessary for the satisfaction of the operational equivalences. We start with the ones pertaining to the pair of input preparations $(a,b)$. For, $x\in\{a,b\}$, let
$$\rho_{x^\perp} := (1-v)\ketbra{x^\perp}{x^\perp}+v\frac{\Id_{ab}}{2},$$
with $\ket{x^\perp}\in  {\rm span}(\{\ket{a},\ket{b}\})$ and $\braket{x}{x^\perp}=0$. Note that these are naturally thought as the noisy version of the perfect orthogonal preparations, $\rho_{x^\perp} = \mathcal{N}_v (\ketbra{x^\perp 0}{x^\perp 0})$. Now, let us check that the operational equivalence is satisfied,
\begin{align*}
\frac{1}{2}\rho_a + \frac{1}{2}\rho_{a^\perp} &= (1-v)\frac{\ketbra{a}{a}+\ketbra{a^\perp}{a^\perp}}{2}+v\frac{\Id_{ab}}{2}=\frac{\Id_{ab}}{2}=(1-v)\frac{\ketbra{b}{b}+\ketbra{b^\perp}{b^\perp}}{2}+v\frac{\Id_{ab}}{2}=\frac{1}{2}\rho_b + \frac{1}{2}\rho_{b^\perp}.
\end{align*}
Next, we consider the pair of preparations $s \in (\alpha,aa)$. Let
\begin{align*}
\rho_{s^\perp}:=(1-v)^2\ketbra{s}{s}+(1-(1-v)^2)\frac{\Id}{4},
\end{align*}
with $\ket{s^\perp}\in  {\rm span}(\{\ket{\alpha},\ket{aa}\})$, and $\braket{s}{s^\perp}=0$. $\rho_{s^\perp}$ can be seen as the state resulting from preparing $\ket{s^\perp}$ and letting the noise channel act for two steps, i.e., $\rho_{s^\perp} = \mathcal{N}_v \circ \mathcal{N}_v(\ketbra{s^\perp}{s^\perp})$. 

We can now see that with this choice of states the operational equivalences are satisfied:
\begin{align*}
\frac{1}{2}\rho_\alpha + \frac{1}{2}\rho_{\alpha^\perp} & = (1-v)^2\frac{\ketbra{\alpha}{\alpha}+\ketbra{\alpha^\perp}{\alpha^\perp}}{2}+(1-(1-v)^2)\frac{\Id}{4}=(1-v)^2\frac{\Id_{\alpha aa}}{2}+(1-(1-v)^2)\frac{\Id}{4}\\
& =(1-v)^2\frac{\ketbra{aa}{aa}+\ketbra{aa^\perp}{aa^\perp}}{2}+(1-(1-v)^2)\frac{\Id}{4}
=\frac{1}{2}\rho_{aa} + \frac{1}{2}\rho_{aa^\perp},
\end{align*}
with $\Id_{\alpha aa}$ the projector over the ${\rm span}(\{\ket{\alpha},\ket{aa}\})$. Following the same argumentation, one can see that for the remaining pairs of preparations $(\beta,bb)$ and $(aa,bb)$, if we let
\begin{align*}
\rho_{\beta^\perp}&:=(1-v)^2\ketbra{\beta^\perp}{\beta^\perp}+(1-(1-v)^2)\frac{\Id}{4},\\
\rho_{bb^\perp}&:=(1-v)^2\ketbra{bb^\perp}{bb^\perp}+(1-(1-v)^2)\frac{\Id}{4},\\
\end{align*}
with $\ket{bb^\perp},\ket{\beta^\perp}\in  {\rm span}(\{\ket{\beta},\ket{bb}\})$ and $\braket{bb}{bb^\perp}=\braket{\beta}{\beta^\perp}=0$, the operational equivalence for $(\beta,bb)$ is satisfied:
\begin{align*}
\frac{1}{2}\rho_\beta + \frac{1}{2}\rho_{\beta^\perp} =\frac{1}{2}\rho_{bb} + \frac{1}{2}\rho_{bb^\perp}.
\end{align*}
And letting
\begin{align*}
\rho'_{aa^{\perp}}&:=(1-v)^2\ketbra{\overline{aa}^\perp}{\overline{aa}^\perp}+(1-(1-v)^2)\frac{\Id}{4},\\
\rho'_{bb^{\perp}}&:=(1-v)^2\ketbra{\overline{bb}^\perp}{\overline{bb}^\perp}+(1-(1-v)^2)\frac{\Id}{4},
\end{align*}
with $\ket{\overline{aa}^\perp},\ket{\overline{bb}^\perp}\in  {\rm span}(\{\ket{aa},\ket{bb}\})$, and $\braket{aa}{\overline{aa}^\perp}=\braket{bb}{\overline{bb}^\perp}=0$, the operational equivalence for $(aa,bb)$ is satisfied:
\begin{align*}
\frac{1}{2}\widetilde{\rho_{aa}} + \frac{1}{2}\rho'_{aa^{\perp}} =\frac{1}{2}\widetilde{\rho_{bb}} + \frac{1}{2}\rho'_{bb^{\perp}}.
\end{align*}
Notice that $\rho_{aa^{\perp}},\rho'_{aa^{\perp}}$ and $\rho_{bb^{\perp}},\rho'_{bb^{\perp}}$ are alternative choices of orthogonal preparations, tailored to each pair of preparations appearing in the operational equivalences.
\subsection{Noise parameter and Error term in the NC bound}
In this subsection, we find the expression for each of the measurement error probabilities appearing in the error term in Eq.~\eqref{eq:ncbound-noisy} as a function of the noise parameter $v$ of the depolarizing channel.

For $x\in\{a,b\}$,
\begin{align*}
1-\epsilon_x &= p(M_x|P_x) = \Tr \left[\left(\left(1-v\right)\ketbra{x}{x}+v\frac{\Id_{ab}}{2}\right)\left(\left(1-v\right)\ketbra{x}{x}+v\frac{\Id_{ab}}{2}\right)\right]\\
&=(1-v)^2+2v(1-v)\Tr\left[\frac{\ketbra{x}{x}}{2}\right]+v^2\Tr\left[\frac{\Id_{ab}}{4}\right]=(1-v)^2+v(1-v)+\frac{v^2}{2}\\
&\implies \epsilon_x=v-\frac{v^2}{2}.
\end{align*}
For $\chi\in\{\alpha,aa\}$,
\begin{align*}
1-\epsilon_\alpha & =  p(M_\alpha|P_\alpha) = \Tr[\rho_\alpha M_\alpha] = 1-\epsilon_{aa} =p(M_{aa}|P_{aa}) = \Tr[\rho_{aa}M_{aa}]\\
&=\Tr \left[\left(\left(1-v\right)^2\ketbra{\chi}{\chi}+\left(1-(1-v)^2\right)\frac{\Id}{4}\right)\left(\left(1-v\right)\ketbra{\chi}{\chi}+v\frac{\Id}{4}\right)\right]\\
&=(1-v)^3+\frac{(1-v)^2 v}{4}+\frac{(1-(1-v)^2)(1-v)}{4}+\frac{(1-(1-v)^2)v}{4}\\
&=\frac{1}{4}(4-9v+9v^2-3v^3)\\
&\implies \epsilon_{aa}=\epsilon_\alpha=\frac{3}{4}v(3-3v+v^2),
\end{align*}
and analogously for $\chi$ in $\{\beta,bb\}$ and $\{aa,bb\}$. Hence,
\begin{align*}
{\rm Err} &= \epsilon_\alpha + \epsilon_\beta + \epsilon_a + \epsilon_b + 2(\epsilon_{aa} + \epsilon_{bb}) = 2(v-\frac{v^2}{2})+6\cdot \frac{3}{4}v(3-3v+v^2)
= \frac{1}{2}v(31-29v+9v^2).
\end{align*}
Following the same arguments, it is easy to see that, for the case of symmetric confusabilities, the error term in Eq.~\eqref{eq:ncbound-noisy} becomes, as a function of $v$,
\begin{align*}
{\rm Err'}=\frac{1}{8}v(31-21v+9v^2)
\end{align*}
\subsection{Quantum performance}
In this last subsection, we compute the global average fidelity $F_g^{\rm Q}$ in the noisy setting of the optimal quantum cloner for the ideal setting as a function of the noise channel's parameter $v$.
\begin{align*}
F_g^{\rm Q} &= \frac{1}{2}\Tr [M_{aa}\rho_\alpha]+\frac{1}{2}\Tr [M_{bb}\rho_\beta]\\
&= \Tr \left[\left(\left(1-v\right)^2\ketbra{xx}{xx}+\left(1-(1-v)^2\right)\frac{\Id}{4}\right)\left(\left(1-v\right)\ketbra{\chi}{\chi}+v\frac{\Id}{4}\right)\right]\\
&= (1-v)^3|\braket{xx}{\chi}|^2+\frac{(1-(1-v)^2)(1-v)}{4}+\frac{(1-v)^2 v}{4}+ \frac{(1-(1-v)^2)v}{4}\\
&=(1-v)^3|\braket{xx}{\chi}|^2+\frac{1}{4}v(3-3v+v^2).
\end{align*}
Hence,
$$F_g^{\rm Q,noisy}(v):=(1-v)^3F_g^{\rm Q,opt}+\frac{1}{4}v(3-3v+v^2).$$

\end{document}